   \documentclass[final,1p,times]{elsarticle}

\usepackage{amsmath, amsthm, amscd, amsfonts, amssymb, graphicx, color}
\usepackage[bookmarksnumbered, colorlinks, plainpages]{hyperref}
\usepackage{wrapfig}
\usepackage{subfig}

\newdefinition{rmk}{Remark}
\newproof{pf}{Proof}

\begin{document}

\newcommand{\be}{\begin{equation}}
\newcommand{\ee}{\end{equation}}
\newcommand{\Om}{\Omega}
\newcommand{\h}{H^1_0(\Omega)}
\newcommand{\lt}{L^2(\Omega)}
\theoremstyle{plain}
\newtheorem{thm}{Theorem}[section]
\newtheorem{cor}[thm]{Corollary}
\newtheorem{lem}[thm]{Lemma}
\newtheorem{prop}[thm]{Proposition}
\theoremstyle{definition}
\newtheorem{defn}{Definition}[section]
\theoremstyle{remark}
\newtheorem{rem}{Remark}[section]
\numberwithin{equation}{section}
\renewcommand{\theequation}{\thesection.\arabic{equation}}
\numberwithin{equation}{section}
\begin{frontmatter}



\title{ A nonlinear eigenvalue problem arising in  a nanostructured quantum dot}
\author[ya]{Abbasali Mohammadi \corref{cor1}}
\ead{mohammadi@yu.ac.ir}
\author[rvt]{Fariba Bahrami}
\ead{fbahram@tabrizu.ac.ir} \cortext[cor1]{Corresponding author}
\address[ya]{Department of Mathematics, College of Sciences,
Yasouj University, Yasouj, Iran, 75914-353 }
\address[rvt]{Faculty of Mathematical Sciences, University of
Tabriz, 29 Bahman St., Tabriz, Iran, 51665-163 }


\begin{abstract}
In this paper we investigate a minimization problem related to the
principal eigenvalue of the $s$-wave  Schr\"{o}dinger operator.
The operator depends nonlinearly on the eigenparameter. We prove
the existence of a solution for the optimization problem and the
uniqueness will be addressed when the domain is a ball. The
optimized solution can be applied to design new electronic and
photonic devices based on the quantum dots.
\end{abstract}


\end{frontmatter}


\section{Introduction}\label{intro}
Quantum dot nanostructures  have attracted broad interest   in the past few years because of their unique
physical properties and potential  applications in micro- and nanoelectronic
devices.
 In such nanostructures, the  free carriers are confined  to a small region of
space by potential barriers. If the size of this region is less than the electron wavelength, the electronic states become
quantized at discrete energy levels.  Due to
the possibility of precise control over the conductivity  by
adjusting the energy levels via the configuration, quantum dot
structures have received tremendous attention from many physicists and
scientists \cite{mas}. The problem of finding
the energy states in these structures is regarded as
 an essential step to study the optical and electrical
properties.

 Motivated by the above explanation, in this paper  we consider a
 nanostructure quantum dot and the Schr\"{o}dinger equation
 governing it.  We  discus an efficient  method that is capable  to predict the
configuration which has a minimum  ground  state energy.

Let us introduce the mathematical equations modeling the structure
and an associated optimization problem. Let $\Omega$ be  a bounded
connected set in $\mathbb{R}^n$ with smooth boundary. Suppose that $p_0$
and $q_0$ are two Lebesgue measurable functions satisfying $0\leq
p_0, q_0\leq h$ in $\Omega$, where $h$ is a positive constant. To
avoid trivial situations, we assume that $p_0$ and $q_0$ are not constant
functions. Define $\mathcal{P}$
 and $\mathcal{Q}$ as the family of all measurable functions which
 are rearrangements of $p_0$ and $q_0$ respectively. For $p \in  \mathcal{P}$ and  $q \in  \mathcal{Q}$, the governing  Hamiltonian equation is the following  $s$-wave Schr\"{o}dinger
 equation, \cite{swave},
\begin{equation}
  \label{mpde}
-\frac{\hbar^2}{2m}\Delta u + q(x)u+ 2\lambda p(x) u= \lambda ^2
u,\quad \mathrm{in}\quad\Omega, \quad u=0, \qquad
\mathrm{on}\quad\partial \Omega,
\end{equation}
where  $\hbar$ stands
for  Planck 's constant, $m$ is the mass of particle, $\lambda$ is
the first eigenvalue (ground state energy) and $u$ is the
corresponding eigenfunction (wave function).

 In Schr\"{o}dinger \eqref{mpde}, the potential function is of the form
 \begin{equation*}
 V(\lambda,x)= q(x)+ 2\lambda p(x),
 \end{equation*}
 where it depends on the ground state energy.
 Let us mention that $\lambda$ can be described as a function of $p$ and $q$. Hence, we use notation $\lambda_{p,q}$ to emphasize its dependence on  $p$ and $q$.

  We seek potentials that minimize the first eigenvalue corresponding to equation \eqref{mpde} relative to $  \mathcal{P}$ and  $  \mathcal{Q}$.
 To determine the
potential  which gives  the minimum  ground  state energy,
we should study the following minimization problem
\begin{equation}
  \label{mmp}
\inf_{p \in  \mathcal{P},q \in  \mathcal{Q} }\lambda_{p,q}.
\end{equation}

 These type of optimization problems for eigenvalues of  linear or nonlinear elliptic partial
 differential equations have been intensively attractive  to mathematicians in the past decades.  They
 have several applications as for instance the stability of vibrating bodies, the propagation of waves
 in composite media and the thermic insulation of conductors; see \cite{henrot} for an overview of the topic.
 However, it should be mentioned that the majority
of the investigated nonlinear models are nonlinear in their differential
operator part \cite{emamipro,cuccu,derlet}.

 Equation \eqref{mpde} can be regarded as a nonlinear elliptic eigenvalue problem such that
 the nonlinearity is originated from the nonlinear dependence on the eigenvalue.
  We note that such systems have been under less
attention in this field of study \cite{abbasali1}.
 In the linear  problems,
 the analysis of the eigenvalues is based essentially on the Rayleigh quotient associated with the eigenvalues.
  For nonlinear eigenvalue problem \eqref{mpde}, we should apply the  Rayleigh functional corresponding
  to the eigenvalue. In this paper we extend rearrangements techniques to find an optimal  eigenvalue
  of a nonlinear problem. This eigenvalue minimization  problem is more difficult than that of the linear problems due
  to the more complicated form of the Rayleigh   functional. We hope this paper would be a motivation to further study in this direction.

 One can find some quantum dot models where the  Schr\"{o}dinger
 equations governing  them are nonlinear with respect to the energy \cite{Wang,yim1,yim2}.

 Our paper is organized as follows. In the next section we review
rearrangement theory with an eye on the optimization problem
\eqref{mmp}. In the third section we derive a formula for the
first eigenvalue of the problem \eqref{mpde}. Then we prove
the existence of a solution to the problem \eqref{mmp}. In the
fourth section, we examine the uniqueness problem and we shall
investigate the configuration of the unique solution. In the last
section, we will give an overview of our results with a numerical
example which shows the  physical significance of the findings.

\section{Preliminaries}\label{pre}
In this section we state some results from the rearrangement theory  related to our optimization problem \eqref{mmp}. The reader can refer to  \cite{bu89,Alvino} 
for further information about the rearrangement theory. In this paper, we denote with $|\mathcal{A}| $ the Lebesgue measure of the measurable set $\mathcal{A}\subset\Bbb{R}^n$.

 Two Lebesgue measurable functions $p:\Om \rightarrow \Bbb{R}$, $p_0:\Om \rightarrow \Bbb{R}$  are said to be rearrangements of each other if
 \begin{equation}\label{rea}
|\{x\in \Om : p(x)\geq \alpha\}|=|\{x\in \Om: p_0(x)\geq
\alpha\}|\qquad~\quad\forall\alpha\in \mathbb{R}.
\end{equation}
The notation $p\sim p_0$ means that $p$ and $p_0$ are rearrangements of each other. Consider $p_0:\Om \rightarrow \Bbb{R}$. The class of rearrangements generated by $p_{0}$, denoted $\mathcal{P}$, is defined as follows
 \begin{equation*}
\mathcal{P}=\{p:p\sim p_{0}\}.
\end{equation*}
Consider a function $q\in L^r(\Om)$, $r\geq 1$. A level set of this function is
\begin{equation*}
\{x\in \Omega:\quad q(x)=\alpha\}, \qquad \alpha \in \Bbb{R}.
\end{equation*}
  Throughout this paper we shall write increasing instead  of non-decreasing, and decreasing instead of non-increasing. The following two lemmas were proved in \cite{bu89}.
\begin{lem}\label{ber1}
 Let $p\in L^r(\Omega)$, $r>1$,  and let $q\in L^s(\Omega),s=r/(r-1)$. Suppose that every level set of $q$  has measure zero. Then, there exists an increasing function $\xi:\Bbb{R}\rightarrow \Bbb{R}$  such that $\xi(q)$ is a rearrangement of $p$. Furthermore, there exists a decreasing function $\eta:\Bbb{R}\rightarrow \Bbb{R}$  such that $\eta(q)$ is a rearrangement of $p$.
\end{lem}
 We denote with $\overline{\mathcal{P} } $ the weak closure of $\mathcal{P}$ in $L^r(\Omega)$.
\begin{lem}\label{ber2}
Let $\mathcal{P}$ be the set of rearrangements of a fixed function
$p_{0}\in L^r(\Omega)$, $r>1$, $p_{0}\not\equiv 0$, and let $q\in
L^s(\Omega)$, $s=r/(r-1)$, $q\not\equiv 0$. If there is an
increasing function $\xi$  such that $\xi(q)\in \mathcal{P}$, then
\begin{equation*}
\int_{\Omega} p q dx \leq \int_{\Omega} \xi(q)q dx~~~~~~~
\qquad\qquad\forall ~p \in\overline{\mathcal{P}},
\end{equation*}
and the function $\xi(q)$ is the unique maximizer relative to
$\overline{\mathcal{P}}$. Furthermore, if there is a decreasing
function $\eta$ such that $\eta(q)\in \mathcal{P}$, then
\begin{equation*}
\int_{\Omega} p q dx \geq \int_{\Omega} \eta(q)q
dx~~~~~~~\qquad\qquad \forall ~p \in\overline{\mathcal{P}},
\end{equation*}
and the function $\eta(q)$ is the unique minimizer relative to
$\overline{\mathcal{P}}$.
\end{lem}

\begin{lem}\label{ber3} Consider the rearrangement class $\mathcal{P}\subset L^r(\Omega)$ generated by $p_0$ and $r\geq1$.
Then  $\|p\|_{L^r(\Omega)}=\|p_0\|_{L^r(\Omega)}$ for every $p\in
\mathcal{P}$.
\end{lem}
\begin{proof}
See \cite{bu89}.
\end{proof}
Let us state here one of the essential tools in  studying rearrangement
optimization problems, see \cite{bu89}.

\begin{lem}\label{maxre}
Let $\mathcal{P}$ be the set of rearrangements of a fixed function
$p_{0}\in L^r(\Omega)$, $r>1$, $p_{0}\not\equiv 0$, and let $g\in
L^s(\Omega)$, $s=r/(r-1)$, $g\not\equiv 0$. Then there exists
$\widetilde{p}$ in $\mathcal{P}$ such that
\begin{equation*}
\int_{\Omega} pg dx\leq \int_{\Omega} \widetilde{p}g dx,
\end{equation*}
for every $p$ in $\mathcal{P}$.

\end{lem}

 Let us note that in this paper for a measurable function $q$ on
 $\Omega$ the strong support (or simply support) of $q$ is
 $\mathrm{supp}(q)=\{x\in \Omega :\quad q(x)>0\}$.  We finish this section with a technical assertion.
\begin{lem}\label{weak}
Let $\{p_k\}^\infty_1$ be a sequence of functions in $L^2(\Om)$ which
converges weakly to $p$ and satisfies
\begin{equation*}
0\leq p_k(x)\leq M,
\end{equation*}
almost everywhere in $\Om$. Then, we have
\begin{equation*}
0\leq p(x)\leq M,
\end{equation*}
almost everywhere in $\Om$.

\end{lem}
\begin{proof}
First we denote $\mathcal{A}=\{x\in \Om: p(x)<0\}$ and show
$|\mathcal{A}|=0$. Taking $q=\chi_\mathcal{A} \in L^2(\Om)$, we have
\begin{equation*}
\int_{\Om} q p_k dx\rightarrow \int_{\Om} q p dx,
\end{equation*}
as $k\rightarrow \infty$. However, the left-hand side is a
non-negative sequence of real numbers and so the right-hand side
should be non-negative which implies $|\mathcal{A}|=0$. Applying
the lower semicontinuity of the $L^\infty$ norm \cite{lieb},
yields $p(x)\leq M$ almost everywhere in $\Om$.
\end{proof}
\section{\bf {\bf \em{\bf Existence result}}}

 This section is devoted to the proof of the existence of a solution for
 problem \eqref{mmp}. To this end, we should propose some
 restriction on $p_0$ and $q_0$. First, let us introduce the new notation   $\gamma=\hbar^2/{2m}$ which will be used hereafter in this paper for simplicity. The
  condition corresponding to  $p_0$ is
\begin{equation}\label{fc}
0\leq p_0(x)<\frac{\sqrt{ \gamma C_\Omega}}{2},
\end{equation}
almost everywhere in $\Om$ such that $C_\Omega$ is the best
(largest) constant in Poincar\'{e}'s inequality.
We need condition \eqref{fc} to prove lemma \ref{levelset} which yields that the level sets of a wave
function have measure zero. This result permit us to invoke lemmas \ref{ber1} and \ref{ber2} from the rearrangement
theory.

The conditions
related to $q_0$ are rather complicated in comparison with \eqref{fc}. The function $q_0$ should be a non-negative
characteristic function such that
\begin{equation}\label{pc}
\int_{\Omega}\widetilde{p}\psi^2dx+\sqrt{(\int_{\Omega}\widetilde{p}\psi^2dx)^2+\int_{\Omega}\widetilde{q}\psi^2dx+\gamma\|\psi\|^2_{H^1_0(\Omega)}}<\sqrt{\|q_0\|_{L^\infty(\Omega)}},
\end{equation}
where $\psi$ is the eigenfunction associated with the principal
eigenvalue of the Laplacian
\begin{equation*}
-\Delta \psi=\lambda \psi\quad \mathrm{in}\quad\Omega, \quad \psi=0 \qquad
\mathrm{on}\quad\partial \Omega,
\end{equation*}
where $\|\psi\|_{L^2(\Omega)}=1$ and $\widetilde{p},
\widetilde{q}$ are the maximizers stated in lemma \ref{maxre} for $g=\psi^2$.
 We need condition \eqref{pc} in the above form  to derive a Rayleigh functional for the
 first eigenvalue, see lemma \ref{lambdalem}. In addition, this condition is necessary to ensure that the level sets
 of an eigenfunction have measure zero. Condition \eqref{pc} generates an interesting physical
 consequence. We say that the energy is confined if
\begin{equation}\label{confc}
V(\lambda,x)<\lambda^2,
\end{equation}
in a subset of $\Omega$.  See \cite{mas} for further information about physical significance of this condition.
This condition ensure that in our optimal  quantum dot the energy is confined, see section \ref{physinter}.

   Let us state the main result of this
section.
\begin{thm}\label{mainthem1}
Assume  properties \eqref{fc} and \eqref{pc} hold, then minimization
problem \eqref{mmp} is solvable. It means that there exist
$\widehat{p}\in \mathcal{P}$ and  $\widehat{q}\in \mathcal{Q}$ such
that
\begin{equation*}
\widehat{\lambda}=\lambda_{\widehat{p},\widehat{q}}= \inf_{p \in
\mathcal{P},q \in  \mathcal{Q} }\lambda_{p,q}.
\end{equation*}
\end{thm}

To establish the main theorem, we need  some preparation. Let us
investigate problem \eqref{mpde} more carefully. There is something
nonstandard in the equation \eqref{mpde}. The equation depends
nonlinearly upon the parameter $\lambda$. Therefore, the eigenvalues
of  \eqref{mpde} cannot be characterized by standard variational
principles like the minimax principle of Poincar\'{e}. Hence, we
should use the generalization of these standard variational
principles to achieve a variational formula representing the
eigenvalues in \eqref{mpde}. Voss \emph{et al}. \cite{vos82,vos2003}
and Turner  \cite{turner}  generalized the standard Poincar\'{e}
minimax characterization to the nonlinear eigenvalue problems with
nonlinear dependence on the eigenvalues. To derive a variational
formula, we assemble the conclusions developed in \cite{vos2003}. In
the sequel, $(\cdot, \cdot)$ denotes the inner product for
$H^1_0(\Omega)$.

Multiplying \eqref{mpde} by $\phi\in H^1_0(\Omega)$ and integrating
by parts one gets the following variational formula
\begin{equation}\label{mpdev}
\gamma\int_{\Omega}\nabla u\nabla\phi dx+ \int_{\Omega}qu\phi
dx+2\lambda\int_{\Omega}pu\phi dx=\lambda^2\int_{\Omega}u\phi dx,
\end{equation}
for every $\phi\in H^1_0(\Omega)$. Fix $u\in H^1_0(\Omega)$,
every summand in \eqref{mpdev} can be viewed as a bounded linear
functional on $H^1_0(\Omega)$. Thanks  to the Riesz representation
theorem, equation \eqref{mpdev} is equivalent to

\begin{equation}\label{mpdev2}
\lambda^2(l_1(u),\phi)-2\lambda (l_2(u),\phi)-(l_3(u),\phi)-\gamma (u,\phi)=0, \qquad
\end{equation}
for all $\phi\in H^1_0(\Omega)$. Since $u$ is arbitrary, we find
\begin{equation}\label{mpdeev}
\lambda^2l_1(u)-2\lambda l_2(u)-l_3(u)-\gamma I(u)=0, \qquad
 \end{equation}
for every $u \in H^1_0(\Omega)$
where $l_i:H^1_0(\Omega)\rightarrow H^1_0(\Omega)$, $i=1..3$, are
bounded linear operators and $I:H^1_0(\Omega)\rightarrow H^1_0(\Omega)$ is the identity operator. In
addition, all operators in \eqref{mpdeev}  are selfadjoint
operators. Hence, we can infer that equation \eqref{mpdeev} is
equivalent to the nonlinear eigenvalue problem
\begin{equation}\label{mpdeev2}
\mathcal{F}(\lambda)u=\lambda^2l_1(u)-2\lambda l_2(u)-l_3(u)-\gamma I(u)=0,
\end{equation}
where $\mathcal{F}:H^1_0(\Omega)\rightarrow H^1_0(\Omega)$, is a
family of selfadjoint and bounded operators for $\lambda\in
J=(0,\sqrt{\|q_0\|_{L^{\infty}(\Om)}})$. In view of \eqref{mpdev}
and \eqref{mpdeev2} we have
\begin{equation}\label{g}
(\mathcal{F}(\lambda)u,u)=\lambda^2\|u\|^2_{L^2(\Omega)}-2\lambda\int_{\Omega}pu^2dx-\int_{\Omega}qu^2dx-\gamma\|u\|^2_{H^1_{0}(\Omega)},
\end{equation}
is continuously differentiable, and  for every fixed $u\in
H^1_{0}(\Omega) \backslash \{0\}$ the equation
\begin{equation}\label{g=0}
(\mathcal{F}(\lambda)u,u)=0,
\end{equation}
has at most one solution in the interval $J$. Accordingly,
equation \eqref{g=0} implicitly defines a functional $\mathcal{R}$ on some
subset $\mathcal{D}$ of $H^1_{0}(\Omega)\backslash\{0\}$ which is
called the Rayleigh functional. The Rayleigh functional $\mathcal{R}$ is
calculated as
\begin{equation}\label{ray formula}
\mathcal{R}(u)=\frac{
\int_{\Omega}pu^2dx+\sqrt{(\int_{\Omega}pu^2dx)^2+(\int_{\Omega}qu^2dx+\gamma\|u\|^2_{H^1_0(\Omega)})\|u\|^2_{L^2(\Om)}
}  }{\|u\|^2_{L^2(\Om)}},
\end{equation}
for every $u$ in $\mathcal{D}$. We should insist that the set
$\mathcal{D}$ is not empty. In view of lemma \ref{maxre} and condition \eqref{pc}, we have
\begin{equation*}
\mathcal{R}(\psi)\leq \int_{\Omega}\widetilde{p}\psi^2dx+\sqrt{(\int_{\Omega}\widetilde{p}\psi^2dx)^2+\int_{\Omega}\widetilde{q}\psi^2dx+\gamma\|\psi\|^2_{H^1_0(\Omega)}}<\sqrt{\|q_0\|_{L^\infty(\Omega)}},
\end{equation*}
which means that $\mathcal{R}(\psi)$ belongs to $J$ and then $\mathcal{D}$ is not empty.

 Let us note that the Rayleigh  functional
is a generalization of the Rayleigh quotient in the theory of linear
eigenvalue problems. Since the Rayleigh functional $\mathcal{R}$ is not
defined on the entire space $H^1_{0}(\Omega)\backslash\{0\}$, then
the eigenproblem \eqref{mpdeev2} is called nonoverdamped. Werner and
Voss, \cite{vos82}, studied the general nonoverdamped case and
proved a minmax principle generalizing the characterization of
Poincar\'{e}.

 Additionally, it can be verified that
\begin{equation}\label{dg>0}
\frac{\partial}{\partial
\lambda}(\mathcal{F}(\lambda)u,u)\mid_{\lambda=\mathcal{R}(u)}=2(\lambda\|u\|^2_{L^2(\Om)}- \int_{\Omega}pu^2dx  )>0,\qquad
\forall u\in \mathcal{D},
\end{equation}
applying \eqref{ray formula}. It remains to examine the existence
of  a function $\zeta(\lambda)>0$ such that the linear operator
$\mathcal{F}(\lambda)+\zeta(\lambda)I$ for every $\lambda \in J$
is
 completely continuous. To achieve  this aim, if we let
$\zeta(\lambda)=\gamma$ then in view of \eqref{mpdeev}, it suffices
to show that $l_i$, $i=1..3$, are completely continuous. We only
derive the last assertion for the case $l_2$. Other cases can be
proved in the same way and the proofs are omitted.

 Recall from \eqref{mpdev} and \eqref{mpdev2} that
 \begin{equation*}
 (l_2(u),l_2(u))=\int_{\Omega}p u l_2(u)dx.
 \end{equation*}
   Consider a weak convergent sequence $\{u_k\}_1^\infty$ where
$u_k\rightharpoonup u$ in $H^1_0(\Omega)$. Then the  compact
embedding of $H^1_0(\Omega)$ into $L^2(\Omega)$, see
\cite{gilbarg}, implies that $\{u_k\}_1^\infty$ converges strongly
to $u$ in $L^2(\Omega)$. On the other hand, using Poincar\'{e}'s
inequality we have
\begin{eqnarray*}
\lim_{k\rightarrow \infty}\|l_2(u_k)-l_2(u)\|^2_{H^1_0(\Omega)}
=\lim_{k\rightarrow \infty}\int_\Omega p(u_k-u)l_2(u_k-u) dx \\
\leq (\frac{1}{C_\Om})\lim_{k\rightarrow
\infty}\|p_0\|_{L^\infty(\Omega)}\|u_k-u\|_{L^2(\Omega)}\|l_2(u_k)-l_2(u)\|_{H^1_0(\Omega)},
\end{eqnarray*}
where this yields  $l_2(u_k)$ converges strongly to $l_2(u)$ in
$H^1_0(\Omega)$. This implies that $l_2$ is completely continuous.
Using the above discussions, it can be said that all assumptions of
theorem $1$ in \cite{vos2003} are satisfied and the following lemma can
be deduced.
\begin{lem}\label{lambdalem}
Suppose  \eqref{pc} holds, then the first eigenvalue of
equation \eqref{mpde} has variational formulation
\begin{equation}\label{lambda}
\lambda= \underset{\| u\|_{L^2(\Omega)}=1}{\underset{u\in
H_{0}^{1}(\Omega)}{\min}}\int_{\Omega}pu^2dx+\sqrt{(\int_{\Omega}pu^2dx)^2+\int_{\Omega}qu^2dx+\gamma\|u\|^2_{H^1_0(\Omega)}}.
\end{equation}

\end{lem}
In the two following lemmas, we examine the eigenfunction of
\eqref{mpde}.
\begin{lem}\label{eigenfunction}
Let $u$ be an eigenfunction corresponding to the first eigenvalue
of \eqref{mpde} then\\
i) $u \in H^2(\Omega)\cap C^{1,\delta}(\Omega)\cap
C(\overline{\Omega})$ for some $\delta \in (0,1)$,\\
ii) $u>0$ in $\Omega$, \\
iii) $u$ is unique up to a constant factor.
\end{lem}
\begin{proof} $\mathrm{i)}$ Equation \eqref{mpde} can be considered as
\begin{equation*}
-\gamma\Delta u +v(x)u=0 ,\quad in\quad\Omega, \quad u=0, \qquad
\mathrm{on}\quad\partial \Omega,
\end{equation*}
where $v(x)=q(x)+ 2\lambda p(x) - \lambda ^2 $. By standard
regularity results for linear elliptic partial differential
equations, see \cite{gilbarg}, the first assertion is obtained.\\
$\mathrm{ii)}$ In view of \eqref{lambda}, we can regard $|u|$ as
 an eigenfunction. Applying Harnack's inequality
\cite{gilbarg}, leads us to the fact
that eigenfunctions associated with $\lambda$ have constant sign.\\
$\mathrm{iii)}$ Let $\widetilde{u}$ be an eigenfunction of
\eqref{mpde} corresponding to $\lambda$. According to part
$(\mathrm{ii})$, we have $\int_{\Omega} \widetilde{u} dx>0$ and so
there exists a real constant $\alpha$ such that
$\int_{\Omega}u-\alpha \widetilde{u} dx= 0$. But since $u-\alpha
\widetilde{u}$ is also a solution of \eqref{mpde} associated with the
first eigenvalue $\lambda$ and $\int_{\Omega}u-\alpha
\widetilde{u} dx= 0$, one  can arrive at $u\equiv \alpha
\widetilde{u}$.

\end{proof}

\begin{lem}\label{levelset}
Suppose relations \eqref{fc} and \eqref{pc} hold and $u$ is a
solution of \eqref{lambda}  normalized as
$\|u\|_{L^2(\Omega)}=1$. The level sets of $u$  have
 measure zero.
\end{lem}
\begin{proof}
Employing lemma \ref{eigenfunction}, we know that $u$ is a positive
function that satisfies
\begin{equation}\label{quadlambda}
-\gamma\Delta u= (\lambda^2-2\lambda p-q)u,
\end{equation}
almost everywhere in $\Omega$.\\
Let us define $\mathcal{A}=\mathrm{\mathrm{supp}}(p)$ and
$\mathcal{B}=\mathrm{\mathrm{supp}}(q)$.  We will show that the
right hand side of \eqref{quadlambda} is never zero in  $\Om$. This
is a clear conclusion in $(\mathcal{A}\cup\mathcal{B})^c$.

 We now claim that $\lambda^2-2\lambda p(x)>0$ in $\mathcal{A}-\mathcal{B}$.
From \eqref{lambda} we deduce $\lambda>\sqrt{\gamma
}\|u\|_{H^1_0(\Om)}$. Since $\psi$ is   the
normalized eigenfunction corresponding to the principal eigenvalue of the Laplacian
with Dirichlet's boundary condition, we have $\|u\|_{H^1_0(\Om)}\geq \|\psi\|_{H^1_0(\Om)}$ and so  $\lambda>\sqrt{\gamma
}\|\psi\|_{H^1_0(\Om)}=\sqrt{\gamma C_{\Om}}$. This yields that $\lambda^2-2\lambda p(x)>0$
because of \eqref{fc}.

 We know that $q(x)$ is a characteristic function due
to \eqref{rea}. Recall from lemma \ref{ber3} that $\|q\|_{L^\infty(\Om)}=\|q_0\|_{L^\infty(\Om)}$. Invoking \eqref{pc}, \eqref{lambda} and lemma \ref{maxre}, we see
\begin{align}
\lambda&\leq&
\int_{\Omega}pu^2dx+\sqrt{(\int_{\Omega}pu^2dx)^2+\int_{\Omega}qu^2dx+\gamma\|u\|^2_{H^1_0(\Omega)}}\\
&\leq&
\int_{\Omega}\widetilde{p}\psi^2dx+\sqrt{(\int_{\Omega}\widetilde{p}\psi^2dx)^2+\int_{\Omega}\widetilde{q}\psi^2dx+\gamma\|\psi\|^2_{H^1_0(\Omega)}}\leq
\sqrt{\|q_0\|_{L^\infty(\Om)}},
\end{align}
which implies  $\lambda^2-q(x)<0$ in $\mathcal{B}$. In summary,
$-\Delta u\neq 0$ almost everywhere in $\Om$. Consequently, the level sets of $u$  have
 measure zero applying lemma 7.7 of
\cite{gilbarg}.
\end{proof}

\begin{rem}
Using lemma \ref{weak}, it is a straightforward conclusion that
lemmas \ref{lambdalem} and  \ref{eigenfunction} are valid when
$q\in {\overline{\mathcal{Q}}}, p\in {\overline{\mathcal{P}}}$
where $ {\overline{\mathcal{P}}}$ and $ {\overline{\mathcal{Q}}}$
are the weak closure of $\mathcal{Q}$ and $\mathcal{P}$ in
$L^2(\Om)$ respectively. In addition, lemma \ref{levelset} is
valid for $p\in {\overline{\mathcal{P}}}$.

\end{rem}

Now we are ready to state the proof of theorem \ref{mainthem1}
\begin{proof}
There exists a real number $\widehat{\lambda}$ and minimizing
sequence $\{p_k\}^\infty_1$ and $\{q_k\}^\infty_1$ such that
\begin{eqnarray*}
\widehat{\lambda}=\inf_{p \in \mathcal{P},q \in  \mathcal{Q}
}\lambda_{p,q}&=& \lim_{k\rightarrow\infty}\lambda_{p_k,q_k} =\\
\lim_{k\rightarrow\infty}
\int_{\Omega}p_ku_k^2dx&+&\sqrt{(\int_{\Omega}p_ku_k^2dx)^2+\int_{\Omega}q_ku_k^2dx+\gamma\|u_k\|^2_{H^1_0(\Omega)}},
\end{eqnarray*}
where $u_k$ is the positive eigenfunction corresponding to
$\lambda_{p_k,q_k}$ normalized such that $\| u_k\|_{L^2(\Omega)}=1$.
Employing lemma \ref{ber3}, we see that the sequences $\{p_k\}^\infty_1$
and $\{q_k\}^\infty_1$ are bounded in $L^\infty(\Omega)$. Hence
there are subsequences (still denoted by $\{p_k\}^\infty_1$ and
$\{q_k\}^\infty_1$) converging  to $\widehat{p}$ and $\widehat{q}$ in
$L^\infty(\Omega)$ with respect to the weak star topology. Moreover,
$\{u_k\}^\infty_1$ is a bounded sequence in $H^1_0(\Omega)$ and
there is a subsequence (still denoted by $\{u_k\}^\infty_1$)
converging weakly to $\widehat{u}$ in $H^1_0(\Omega)$. The
compact embedding of $H^1_0(\Omega)$ into $L^2(\Omega)$ (see
\cite{gilbarg}) yields that $\{u_k\}^\infty_1$ converges strongly to
$\widehat{u}$ in $L^2(\Omega)$. In summary, we have
\begin{equation}\label{pwcon}
p_k\rightharpoonup \widehat{p},\qquad q_k\rightharpoonup
\widehat{q} \quad  \mathrm{in}\quad L^\infty(\Omega),
\end{equation}

\begin{equation}\label{uwcon}
u_k\rightharpoonup \widehat{u},\quad \mathrm{in}\quad H^1_0(\Omega),
\qquad u_k\rightarrow \widehat{u} \quad \mathrm{in}\quad
L^2(\Omega).
\end{equation}
On the other hand, for all $\phi \in H^1_0(\Omega)$ we have
\begin{equation}\label{integralcon}
\int_{\Omega} p_k u_k \phi dx \rightarrow \int_{\Omega} \widehat{p}
\widehat{u} \phi dx, \qquad \int_{\Omega} q_k u_k \phi dx
\rightarrow \int_{\Omega} \widehat{q} \widehat{u} \phi dx,
\end{equation}
since, for instance,  applying lemma \ref{ber3} we see that

\begin{eqnarray*}
\lim_{k\rightarrow \infty}\left| \int_{\Omega} p_k u_k \phi dx-
\widehat{p}\widehat{u} \phi dx \right|&=& \lim_{k\rightarrow
\infty}\left|\int_{\Omega} p_k u_k \phi - p_k \widehat{u }\phi+  p_k
\widehat{u }\phi- \widehat{p}\widehat{u} \phi dx\right|\\ &\leq&
\lim_{k\rightarrow \infty}\|p_0\|_{L^\infty(\Omega)}\int_{\Omega}
\left|(u_k-\widehat{u}) \phi \right|dx+\lim_{k\rightarrow
\infty}\left|\int_{\Omega}\widehat{u}\phi(p_k-\widehat{p})dx\right|,
\end{eqnarray*}
where the right hand side converges to zero, applying (\ref{pwcon})
and (\ref{uwcon}). At last, by means of the continuity of the
integral $\int_{\Omega}\nabla u \nabla\phi$ with respect to $u$ in
$H^1_0(\Omega)$ together with (\ref{integralcon}) we have
\begin{equation*}
\gamma\int_{\Omega}\nabla\widehat{u}\nabla\phi dx+
\int_{\Omega}\widehat{q}\widehat{u}\phi
dx+2\widehat{\lambda}\int_{\Omega}\widehat{p}\widehat{u}\phi
dx=\widehat{\lambda} ^2\int_{\Omega}\widehat{u}\phi dx,
\end{equation*}
for every $\phi$ which belongs to $H^1_0(\Omega)$. Therefore,
$\widehat{\lambda}$ is an eigenvalue of (\ref{mpde}) with
$\widehat{u}$ as its associated eigenfunction corresponding to
$\widehat{p}$ and $\widehat{q}$. In other words,
\begin{eqnarray}\label{lambdahat}
\widehat{\lambda}=\inf_{p \in \mathcal{P},q \in  \mathcal{Q}
}\lambda_{p,q}=\lambda_{\widehat{p},\widehat{q}},
\end{eqnarray}

 It remains to  show that $\widehat{p}\in \mathcal{P}$ and
$\widehat{q}\in  \mathcal{Q}$. Consider the set
\begin{eqnarray*}
\mathcal{C}=\left\{q\in L^2(\Om):\: 0\leq q \leq
\|q_0\|_{L^\infty(\Om)},\, \int_\Om qdx=\int_\Om q_0dx\right\}.
\end{eqnarray*}
Observe that $\widehat{q}\in \mathcal{C}$ in view of lemmas
\ref{ber3} and \ref{weak}.
Consider the following minimization problem
\begin{eqnarray}\label{qmin}
\inf_{q \in  \mathcal{C} }\int q \widehat{u}^2 dx.
\end{eqnarray}
Employing the bathtub principle \cite{lieb}, we can find that there is a characteristic
function in $\mathcal{C}$ where it is a solution of \eqref{qmin}. The minimizer which still denoted $\widehat{q}$
has the form $\widehat{q}=\beta \chi_{\mathcal{D}_q}$ so that
$\beta=\|q_0\|_{L^\infty(\Omega)}$  and
$|\mathcal{D}_q|=|\mathrm{\mathrm{supp}}(q_0)|$.
 Invoking formula
\eqref{lambda}, we see that the new $\widehat{q}$ satisfies equations
\eqref{lambdahat} and $\widehat{q} \in \mathcal{Q}$ based upon
\eqref{rea}.

Next we assert that $\widehat{p}\in \mathcal{P}$. Utilizing lemmas
\ref{eigenfunction} and \ref{levelset}, we can see that
$\widehat{u}^2$ has all level sets with measure zero and
$\widehat{u}^2\in L^1(\Omega)$. Remembering lemmas \ref{ber1} and
\ref{ber2}, there exists a decreasing function $\eta_1$ where
\begin{equation}\label{eta1}
\int_{\Omega} p \widehat{u}^2 dx \geq \int_{\Omega}
\eta_1(\widehat{u}^2)\widehat{u}^2 dx~~~~~~~\qquad\qquad \forall ~p
\in\overline{\mathcal{P}}.
\end{equation}
Note that $\eta_1(\widehat{u}^2)\in \mathcal{P} $ and it is the
unique minimizer in the above inequality. Using relations
\eqref{lambda} and \eqref{eta1}, we obtain
\begin{eqnarray*}
\widehat{\lambda}&\geq&
\int_{\Omega}\eta_1(\widehat{u}^2)\widehat{u}^2dx+\sqrt{(\int_{\Omega}\eta_1(\widehat{u}^2)\widehat{u}^2dx)^2+\int_{\Omega}\widehat{q}\widehat{u}^2dx+\gamma\|\widehat{u}\|^2_{H^1_0(\Omega)}}\\
&\geq& \widehat{\lambda},
\end{eqnarray*}
where by the uniqueness of the minimizer stated above we have
\begin{eqnarray*}
\widehat{p}=\eta_1(\widehat{u}^2),
\end{eqnarray*}
 which asserts that $\widehat{p}$ is in $\mathcal{P}$.

\end{proof}
%
%

\section{\bf {\bf \em{\bf Uniqueness result and shape configuration}}}
  From the  physical point of view, it is important   to know the uniqueness of functions $\widehat{p}(x)$ and $\widehat{q}(x)$  and the
 shape of the potential function $\widehat{V}=\widehat{q}(x)+2\widehat{\lambda} \widehat{p}(x)$. Such questions have been addressed in
 \cite{henrot,chanillo,abbasali2}.   We assume, hereafter, that $p_0$ is a characteristic function and  $\Om$  is a ball centered at the
 origin. Let us recall that $q_0$ is a characteristic function as
 well. For characteristic functions $p_0$ and $q_0$, we can determine the optimizers $\widehat{p}$ and $\widehat{q}$ exactly.

Let us recall some important functions which belong to the
rearrangement class in the case where $\Omega$ is a ball centered at the
origin. Assume $f: \Omega\rightarrow \mathbb{R}$  is a Lebesgue
measurable function then we denote by $f^*: \Omega\rightarrow
\mathbb{R}$ and $f_*: \Omega\rightarrow \mathbb{R}$ the Schwarz
decreasing and increasing rearrangements of $f$ respectively. It
means that $f^*$ and $f_*$ are rearrangements of $f$ such that
$f^*$ is a radial decreasing function, whereas $f_*$ is a radial
increasing function \cite{henrot,hardy}. Next we state some well known rearrangement
inequalities.
\begin{lem}\label{hardy} Suppose $\Omega$ is a ball centered at the origin in $\mathbb{R}^n$. Then
\begin{equation*}
\int_{\Omega} f^*g_* dx\leq\int_{\Omega}fg dx\leq \int_{\Omega}
f^*g^* dx,
\end{equation*}
where $f$ and $g$ are non-negative measurable functions.
\end{lem}
\begin{proof}
See \cite{hardy}.
\end{proof}
\begin{lem}\label{zeimer} Suppose $\Omega$ is a ball centered at the origin in
$\mathbb{R}^n$ and $r>1$. Consider a non-negative function $u \in
W_0^{1,r}(\Omega)$ then \\
$i)$ $u^* \in W_0^{1,r}(\Omega) $ and
\begin{equation*}
\int_{\Omega}|\nabla u|^r dx \geq \int_{\Omega}| \nabla u^*|^r dx,
\end{equation*}
$ii)$ If in the last inequality, the equality holds, and the set
$\{x \in \Omega: \nabla u^*(x)=0,\quad0<u^*(x)<M\}$,
$M=\mathrm{ess} \underset{\Omega} \sup\;u(x)$, has zero measure,
then $u=u^*$.
\end{lem}
\begin{proof}
See \cite{brothers}.
\end{proof}

The main purpose of this section is to prove the following theorem.

\begin{thm}\label{mainthem2}
Assume $\Omega$ is a ball centered at the origin, \eqref{fc} and
\eqref{pc} hold. In addition, $\widehat{p}$ and $\widehat{q}$ are
optimal solutions of \eqref{mmp}. Then these solutions are unique.
Indeed,
\begin{equation*}
\widehat{p}={p_0}_*,\qquad \widehat{q}={q_0}_*,
\end{equation*}
almost everywhere in $\Omega$.
\end{thm}
\begin{proof}
Let $\widehat{u}$ be an eigenfunction corresponding to
$\widehat{\lambda}=\lambda_{\widehat{p},\widehat{q}}$ then
\begin{eqnarray*}
\lambda_{\widehat{p},\widehat{q}}&=&\int_{\Omega}\widehat{p}\widehat{u}^2dx+\sqrt{(\int_{\Omega}\widehat{p}\widehat{u}^2dx)^2+\int_{\Omega}\widehat{q}\widehat{u}^2dx+\gamma\|\widehat{u}\|^2_{H^1_0(\Omega)}}\\
&\geq& \int_{\Omega}{\widehat{p}_*}(\widehat{u}^{*})^{2}
dx+\sqrt{(\int_{\Omega}\widehat{p}_*(\widehat{u}^{*})^{2}dx)^2+\int_{\Omega}\widehat{q}_*(\widehat{u}^{*})^{2}dx+\gamma\|\widehat{u}^{*}\|^2_{H^1_0(\Omega)}}\\
&\geq& \lambda_{\widehat{p},\widehat{q}}.
\end{eqnarray*}
The above inequalities are written by lemma \ref{hardy} and lemma
\ref{zeimer}. This leads us to the equality
\begin{eqnarray}\label{equality}
\int_{\Omega}\widehat{p}\widehat{u}^2dx&+&\sqrt{(\int_{\Omega}\widehat{p}\widehat{u}^2dx)^2+\int_{\Omega}\widehat{q}\widehat{u}^2dx+\gamma\|\widehat{u}\|^2_{H^1_0(\Omega)}}\\
= \int_{\Omega}{\widehat{p}_*}(\widehat{u}^{*})^{2}
dx&+&\sqrt{(\int_{\Omega}\widehat{p}_*(\widehat{u}^{*})^{2}dx)^2+\int_{\Omega}\widehat{q}_*(\widehat{u}^{*})^{2}dx+\gamma\|\widehat{u}^{*}\|^2_{H^1_0(\Omega)}}\nonumber,
\end{eqnarray}
which yields
\begin{eqnarray}\label{equality}
\int_{\Omega}\widehat{p}\widehat{u}^2dx=\int_{\Omega}{\widehat{p}_*}(\widehat{u}^{*})^{2}&,&\quad\int_{\Omega}\widehat{q}\widehat{u}^2dx=\int_{\Omega}\widehat{q}_*(\widehat{u}^{*})^{2}dx,\\\|\widehat{u}\|^2_{H^1_0(\Omega)}&=&\|\widehat{u}^{*}\|^2_{H^1_0(\Omega)}.
\end{eqnarray}

We claim that $\widehat{u}(x)=\widehat{u}^*(x)$. To this end, we
will apply lemma \ref{zeimer} part $(\mathrm{ii})$. Therefore, it
should be shown that
\begin{eqnarray*}
\mathcal{A}=\{x\in \Omega:\quad \nabla\widehat{u}^*=0,\quad
0<\widehat{u}^*<M\},
\end{eqnarray*}
where $M=\|\widehat{u}\|_{L^{\infty}(\Om)}$, has zero Lebesgue
measure. Since $\widehat{p}_*$ and $\widehat{q}_*$ are two
characteristic functions, it should be mentioned that they must
actually have the form
\begin{eqnarray*}
\widehat{p}_*&=&\beta_p \chi_{\mathcal{D}_p},\qquad
\mathcal{D}_p=\{x \in \Om
:\widehat{u}^*(x)\leq t_p \},\\
\widehat{q}_*&=&\beta_q \chi_{\mathcal{D}_q},\qquad
\mathcal{D}_q=\{x \in \Om :\widehat{u}^*(x)\leq t_q \}.
\end{eqnarray*}
 where $t_p, t_q>0$. Suppose $|\mathcal{A}|>0$ and consider the following equation
\begin{eqnarray*}
-\gamma\Delta \widehat{u}^* + \widehat{q}_*(x)\widehat{u}^*+
2\lambda \widehat{p}_*(x) \widehat{u}^*= \lambda ^2
\widehat{u}^*,\quad \mathrm{in}\quad\Omega^\prime=\{x \in \Om
:\widehat{u}^*(x)\neq t_p,t_q\}.
\end{eqnarray*}
The region $\Omega^\prime$ is the union of three open connected
components  whose components have either one or two
rotationally symmetric surfaces as its boundary. In all components
of $\Omega^\prime$, $-\gamma\Delta \widehat{u}^* +
\widehat{q}_*(x)\widehat{u}^*+ 2\lambda \widehat{p}_*(x)
\widehat{u}^*= \lambda ^2 \widehat{u}^*$ is an elliptic equation
with real analytic coefficients and the function on its right hand
side is  real analytic too . Employing the analyticity
theorem \cite{john} yields that $\widehat{u}^*$ is a real
analytic function in every component. Hence, setting
$w=\frac{\partial}{\partial x_i} \widehat{u}^*$ where $i$ is an
integer, $1\leq i\leq n$, leads us to the equation
\begin{eqnarray*}
-\gamma\Delta w+ \widehat{q}_*(x)w+ 2\lambda \widehat{p}_*(x)
w= \lambda ^2 w,\quad \mathrm{in}\quad\Omega^\prime.
\end{eqnarray*}
Since $|\mathcal{A}|>0$ then in a component of $\Omega^\prime$
which we denote it by $\Om_1$, $w$ is zero in a set of positive
measure. For a real analytic function $w$ whose domain is a
connected open set $\Om_1$ we have either $w^{-1}\{0\}=\Om_1$ or
$|w^{-1}\{0\}|=0$ \cite{federer}. This implies that $w\equiv
0$ on $\Om_1$. Accordingly, $\widehat{u}^*$ is constant in the
$x_i$-direction on $\Om_1$. If $\Om_1$ has two rotationally
symmetric surfaces as its boundary, then $\widehat{u}^*$ attains
two different values of the set $\{0,t_p,t_q\}$ on the boundary
which leads us to a contradiction in view of its values in the
$x_i$-direction and continuity of it on $\overline{\Om}$. If
$\Om_1$ has a rotationally symmetric surface as its boundary, then
$\widehat{u}^*$ is a constant function on $\Om_1$ which is a
contradiction to lemma \ref{levelset}. These contradictions
establish the above claim.

 We have proved that $\widehat{u}(x)=\widehat{u}^*(x)$. Recall that level
 sets of $\widehat{u}(x)$ have  zero measure by lemma
 \ref{levelset} and there exists a decreasing function $\eta$ where
\begin{equation*}
\int_{\Omega} \eta(\widehat{u}^{2})\widehat{u}^{2} dx \leq
\int_{\Omega} p \widehat{u}^{2} dx, ~~~~~~~\qquad\qquad \forall ~p
\in\mathcal{P},
\end{equation*}
and the function $\eta(\widehat{u}^{2})$ is the unique minimizer
such that $\eta(\widehat{u}^{2})\in\mathcal{P}$ applying lemma
\ref{ber2}. On the other hand, $\widehat{u}^{2}$ is a radial
decreasing function and then $\eta(\widehat{u}^{2})$ is an
increasing radial function in the rearrangement class
$\mathcal{P}$. Indeed, $\eta(\widehat{u}^{2})={p_0}_*$ and we have
\begin{equation}\label{p0}
\int_{\Omega} {p_0}_*\widehat{u}^{2} dx \leq \int_{\Omega} p
\widehat{u}^{2} dx, ~~~~~~~\qquad\qquad \forall ~p \in\mathcal{P},
\end{equation}
for ${p_0}_*$ as the unique minimizer. Similarly, the relation
\begin{equation}\label{q0}
\int_{\Omega} {q_0}_*\widehat{u}^{2} dx \leq \int_{\Omega} q
\widehat{u}^{2} dx, ~~~~~~~\qquad\qquad \forall ~q \in\mathcal{Q},
\end{equation}
holds for ${q_0}_*$ as the unique minimizer. Finally, \eqref{p0} and
\eqref{q0} leads us to
\begin{equation*}
\widehat{p}={p_0}_*,\qquad \widehat{q}={q_0}_*,
\end{equation*}
the uniqueness assertion.

\end{proof}


\section{Physical Interpretations}\label{physinter}

In this section we will give an overview of our results in the
previous sections with an eye on their  physical importance.
 Let us start with
a numerical example whose  results coincide with findings in
sections 3 and 4.

\begin{figure}[h]\label{ffig}
  \centering
  \subfloat[quantum dot]{\label{ffig:q}\includegraphics[width=0.33\textwidth]{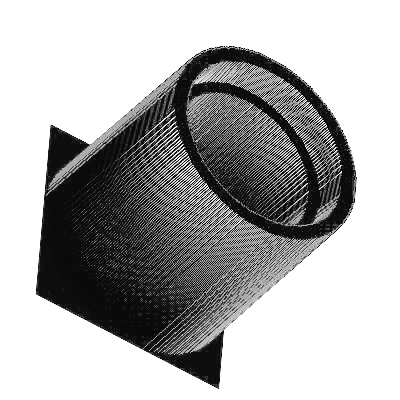}}
  \subfloat[wave function]{\label{ffig:w}\includegraphics[width=0.4\textwidth]{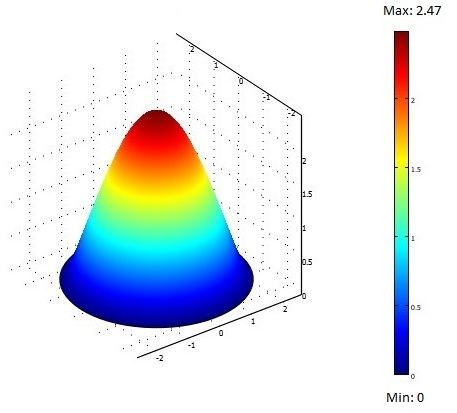}}
    \caption{Wave function and its corresponding quantum dot}
  \label{ffig:animals}
\end{figure}

Take  $\Omega$ to be a ball centered at the origin with radius
$R=2.4$ nm in $\mathbb{R}^2$. Let $q_0(x)$ be a characteristic
function (step function)  equals $2.13 $ eV in a subset
of $\Omega$ with area $3.85$ $\mathrm{nm}^2$ and is zero
elsewhere. Similarly, $p_0(x)$ is a characteristic function equals
$0.27$ eV in a subset of $\Omega$ with area $2$ $\mathrm{nm}^2$ and
is zero elsewhere. Assume the particle's mass is $m=7.81638\times
10^{-32}$ kg, then we have $\gamma=\hbar^2/2m=7.114043325\times
10^{-38}\;\mathrm{(J.s)^2/kg}$. Recall from \eqref{pc} that $\psi$ is the
normalized eigenfunction corresponding to the principal eigenvalue of the Laplacian
with Dirichlet's boundary condition.
It is well known that when $\Om$
is a ball centered at the origin, $\psi=\psi(r)$ is a radial
function and in our case can be calculated as follows:
\begin{equation*}
\psi(r)=4.528173484\times10^{8}J_0( 1.002010649\times10^9r),
\end{equation*}
where $J_0$ is the Bessel function of the first kind of order
zero \cite{henrot}. Now, by a simple calculation, it can be verified that $p_0$
and $q_0$ satisfy conditions \eqref{fc} and \eqref{pc}. Applying
theorems \ref{mainthem1} and \ref{mainthem2} we can deduce that
optimal solution of problem \eqref{mmp} is the first eigenvalue of
\eqref{mpde} corresponding to ${p_0}_*$ and ${q_0}_*$. In other words,
the best potential function is $V={q_0}_*+2\widehat{\lambda}{p_0}_*$.
Recall that ${p_0}_*$ and ${q_0}_*$ are the Schwarz
 increasing rearrangements of $p_0$ and $q_0$ respectively. This means that they
 are two radial characteristic  functions with two circular annular regions as their supports.
Relation \eqref{rea} leads us to the fact that   ${q_0}_*$ is a
step function with the height equals  the height of $q_0$ in an
annulus with outer radius $R$ and inner radius $2.13$ nm and
${p_0}_*$ is a step function with the height equals the height of
$p_0$  in an annulus with outer radius $R$ and smaller radius
$2.26$ nm. Denote by $h_p$ and $h_q$ the heights of ${p_0}_*$ and ${q_0}_*$ respectively. Then, the
radial potential $V={q_0}_*+2\widehat{\lambda}{p_0}_*$ is
\begin{equation*}
 V(r)=
 \left\{
     \begin{array}{ll}
      0\quad\;\; \;\; \quad \quad \quad 0<r\leq r_1,
       \\ h_q \quad \quad \quad \quad \;\; \;     r_1<r\leq r_2,
            \\ h_q+2\widehat{\lambda} h_p \quad \quad       r_{2}<  r \leq R,
            \end{array}
   \right.
 \end{equation*}
where $r_1=2.13$ $\mathrm{nm}$, $r_2=2.26$ $\mathrm{nm}$ and $R=2.4$ $\mathrm{nm}$. Using this potential,
one can construct a three dimensional nanostructure on $\Om$ which is made of two concentric cylinder nested within each other. These cylinders
have heights $h_q$ and $h_q+2\widehat{\lambda} h_p$. When the free carrier is trapped within this structure, the structure
can be considered as a quantum dot \cite{mas}. The optimal quantum dot  is shown in figure \ref{ffig:q}
schematically.

 Inserting ${p_0}_*$ and ${q_0}_*$ into \eqref{mpde}, one can determine the optimal ground
state energy employing the standard finite element Galerkin method
which yields $\widehat{\lambda}^2=0.45$ $(\mathrm{eV})^2$ for this typical
example.  The wave function corresponding to the minimum ground state energy  is illustrated
in schematic figure \ref{ffig:w}.

 We can see that $\lambda^2$ is
less than potential  $V$  when $r_1\leq r \leq R$. This means that in our quantum dot   the energy is confined.

  \textbf{Acknowledgement.}
 The authors would like to express their  deep gratitude to anonymous referees for
helpful comments and useful suggestions.
  We would like to acknowledge
 Professor Heinrich Voss for his thorough  reading of the
manuscript that help us to improve the presentation of the paper.

\end{document}